\documentclass[conference,letterpaper]{IEEEtran}
\usepackage{moreverb}
\usepackage{epsfig}
\usepackage{amsmath,amssymb,amsthm,mathrsfs,amsfonts,dsfont}
\usepackage{adjustbox,lipsum}
\usepackage{algorithm,algorithmic}
\usepackage{amsfonts}
\usepackage{epsfig}
\usepackage{amssymb}
\usepackage{amsmath}
\usepackage{amsthm}
\usepackage{multirow}
\usepackage{setspace}
\usepackage{rotating}
\usepackage{graphicx}
\usepackage{tabularx}
\usepackage{array}
\usepackage{anyfontsize}
\usepackage{color,soul}
\usepackage{graphicx,dblfloatfix}
\usepackage{epstopdf}
\usepackage{blindtext}
\usepackage{amsmath}
\usepackage{amsthm,amssymb,amsmath,bm}
\usepackage{subfig}
\usepackage{amsfonts}
\usepackage{epsfig}
\usepackage{amssymb}
\usepackage{amsmath}
\usepackage{cite}
\usepackage{graphicx}
\usepackage{fancyhdr}
\usepackage[font={small}]{caption}
\usepackage{tabularx}
\usepackage{tcolorbox}
\usepackage{cite}
\usepackage{setspace}

\newtheorem{theorem}{Theorem}

\allowdisplaybreaks
\begin{document}
	
	\title{Moment Generating Function of the AoI in Multi-Source Systems with Computation-Intensive Status Updates  
	\author{  \IEEEauthorblockN{Mohammad Moltafet and  Markus Leinonen  }
	\IEEEauthorblockA{Centre for Wireless Communications -- Radio Technologies\\
		University of Oulu, Finland\\
		e-mail: $\{$mohammad.moltafet, markus.leinonen$\}$@oulu.fi}	
	\and
	\IEEEauthorblockN{  Marian Codreanu}
	\IEEEauthorblockA{Department of Science and Technology\\
		Link\"{o}ping University, Sweden\\
		e-mail: marian.codreanu@liu.se}
}
}

	\maketitle

\begin{abstract}
	We consider a multi-source status update system in which status updates are transmitted as packets containing the measured value of the monitored process and a time stamp representing the time when the sample was generated. The packets of each source are generated according to the Poisson	process and the packets are served according to an exponentially
	distributed service time. 	
	We assume that the received status update packets need further processing before being used (hence, computation-intensive). This is mathematically modeled by introducing an additional server at the sink node.	
	The sink server serves the packets according to an exponentially distributed service time. 
	We introduce two packet management policies, namely, i) a preemptive policy and ii) a blocking policy  and derive the moment generating function (MGF) of the AoI of each source under both policies.  
	In the preemptive policy, a new arriving packet preempts any possible packet that is currently under service regardless of the packet's source index.
	In the blocking policy, when a server  is busy at the arrival instant of a packet the arriving packet is blocked and cleared. We assume that the same preemptive/blocking policy is employed in both transmitter and sink servers. Numerical results are provided to assess the results.
	
	
\end{abstract}

 	\section{Introduction}\label{Introduction}

 Freshness of  status information of various physical processes
 is a key performance enabler in many time-critical applications of wireless sensor networks (WSNs),
 e.g., surveillance in smart home systems, remote surgery,  intelligent
 transportation systems, and drone control. 
 Recently, the age of information (AoI) was proposed as a destination-centric  metric   to measure the information freshness  in status update systems \cite{5984917,6195689,6310931}.
 A status update packet contains the measured value of a monitored process and a time stamp representing the time when the sample was generated. Due to wireless channel access, channel errors, and fading etc., communicating a status update packet through the network experiences a random delay.
 If at a time instant $t$, the most recent status update available at the sink contains the time stamp $U(t)$, AoI is defined
 as the random process $\Delta(t)=t-U(t)$.
 Thus, the AoI measures for each sensor the time elapsed since the last available status update packet was generated at the sensor.

 The  work in \cite{8469047} introduced a powerful technique, called \textit{stochastic hybrid systems} (SHS), to calculate the average AoI.
 %
 In \cite{9103131}, the authors extended the SHS analysis to calculate the moment generating function (MGF) of the AoI. 
 The SHS technique has  been used to analyze the AoI in various queueing models \cite{8437591,8406966,9013935,9048914,moltafet2020sourceawareage,Moltafet2020mgf,9162681,9174099}.
 

 In this paper, we consider a multi-source  system with computation-intensive status updates in which the embedded information in each packet is not available until being processed by a server  at the sink.
 %
  An autonomous driving system in which images are status updates could be an example of the considered system model. In this system, further processing is needed at the sink to expose the embedded status information in the images.
  We introduce two packet management policies,  namely, i)~a preemptive policy and ii)~a blocking policy, and derive the MGF of the AoI of each source under the policies.  According to the  preemptive policy, a new arriving packet preempts any possible packet that is currently under service regardless of the packet's source index.
 According to the blocking policy, when a server  is busy at the arrival instant of a packet the arriving packet is blocked and cleared. We assume that the same preemptive/blocking policy is employed in both transmitter and sink servers. Numerical results are provided to assess the results.
 
  The most related work to this paper is \cite{9001219} where the authors derived the average AoI for a single-source  system where the computation-intensive status updates are generated according to the zero-wait policy. 
 According to the zero-wait policy, a new packet is generated immediately after the previous one is  available at the sink. 
Whereas  \cite{9001219} derives  the average AoI for a single-source setup, we consider multiple sources  and derive the MGF of the AoI. In addition, \cite{9001219} considers the zero-wait policy whereas, we consider random arrivals.

 \section{System Model}\label{System Model}
 We  consider a status update system consisting of two independent sources\footnote{
 Since we have multiple sources of Poisson arrivals, to evaluate the AoI of one source, we can consider two sources without loss of generality.}, one transmitter server, one sink  server, and one  sink, as depicted in Fig. \ref{computationmodel}.
 Each source observes a random process at random time instants. The  sink is interested in timely information about the status of these random processes. Status updates are transmitted as packets, containing the measured value of the monitored process and a time stamp representing the time when the sample was generated. We assume that  the packets  of sources  1 and 2   are generated according to the Poisson process with rates  $\lambda_1$ and $\lambda_2$, respectively, and      
 the packets are served by the transmitter server according to an exponentially distributed service time with mean ${1}/{\mu}$. We assume that the embedded information in each packet is revealed to the sink only after being processed by the sink server which serves the packets according to an exponentially distributed service time with mean ${1}/{\alpha}$.
 
 Let  $\rho_1={\lambda_1}/{\mu}$ and $\rho_2={\lambda_2}/{\mu}$ be the  load of  source 1 and 2, respectively. 
 The  packet generation in the system follows the Poisson process with rate
 $\lambda=\lambda_1+\lambda_2$, and the overall load at the transmitter server is
 $\rho=\rho_1+\rho_2={\lambda}/{\mu}$.   

	\begin{figure}
	\centering
	\includegraphics[width=0.9\linewidth,trim = 0mm 0mm 0mm 0mm,clip]{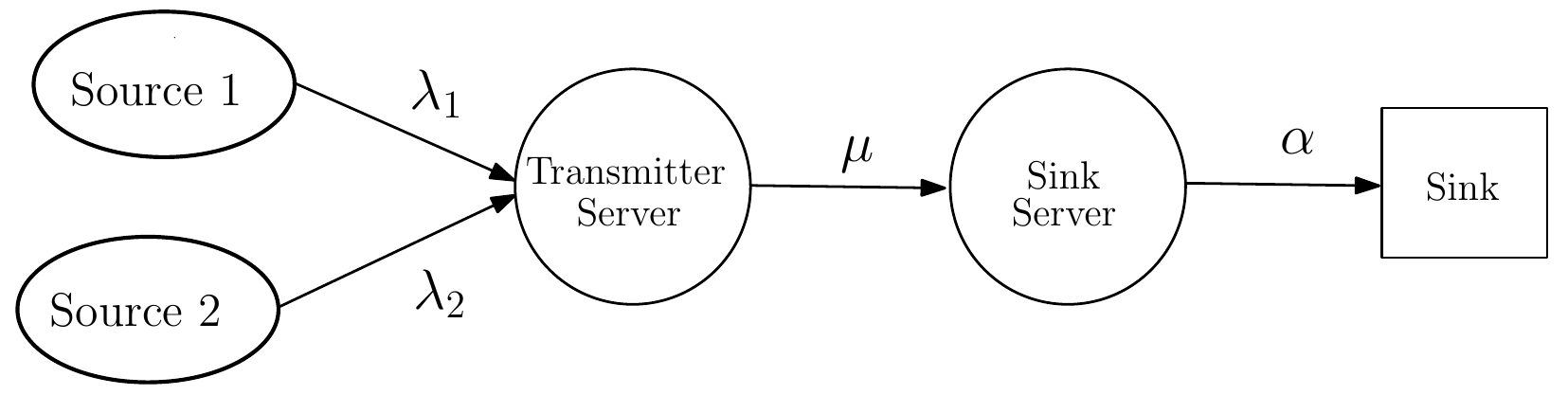}		
	\caption{The considered status update system. 
	 }  
	\label{computationmodel}
			\vspace{-6mm}
\end{figure}
\subsection{Packet Management Policies} \label{SGP-SGP Packet Management Policy}
%
%
%
According to the  preemptive policy, a new arriving packet preempts any possible packet that is currently under service regardless of the packet's source index.
According to the blocking policy, when a server  is busy at the arrival instant of a packet the arriving packet is blocked and cleared. We assume that the same preemptive/blocking policy is employed in both transmitter and sink servers.

\subsection{Summary of the Main Results}
In this paper, we derive the MGF of the AoI for each source under both considered policies which are summarized by the following two theorems. 
\begin{theorem}\label{cMGFtheo1}
	The MGF of the  AoI of source 1 under the preemptive policy is given as
	\begin{align}\label{cmgf1}
		M_{\Delta_1}(s)=\dfrac{\rho_1\bar\alpha(\bar s- \bar\alpha - \rho - 1)}{\bar s(1 - \bar s)(\rho - \bar s)(\rho + 1 - \bar s)+{\bar \alpha}^2\psi_1+\bar \alpha\psi_2},
	\end{align}
	where $\bar s=s/\mu$, $\bar \alpha=\alpha/\mu$, and 
	\begin{align}\label{psi1psi}
	&\psi_1=\rho \bar s -\rho_1+\bar s(1-\bar s),\\&\nonumber
	\psi_2=\rho_1^2(\bar s-1)+\rho_1\rho_2(2\bar s-1)+\rho_1(-1+4\bar s-3{\bar s}^2)+\\&\nonumber~~~~~~~\rho_2\bar s(\rho_2+2-3\bar s)+\bar s(1-3\bar s+2{\bar s}^2).
	\end{align}
\end{theorem}
\begin{proof}
	The proof of Theorem \ref{cMGFtheo1} appears in  Section \ref{Source-Agnostic Preemptive Policy}.
\end{proof}

\begin{theorem}\label{cMGFtheo2}
	The MGF of the  AoI of source 1 under the blocking policy is given by \eqref{cmgf2} (shown at the top of the next page),
		\begin{figure*}
	\begin{align}\label{cmgf2}
		M_{\Delta_1}(s)=\dfrac{\rho_1{\bar \alpha}^2(\bar s-\bar \alpha-\rho)(\rho+1-\bar s)(\bar \alpha+\rho+1-\bar s)}{(\bar \alpha-\bar s)(1-\bar s)(\bar \alpha+\rho)(\rho+1)(\bar s(1 - \bar s)(\rho - \bar s)(\rho + 1 - \bar s)+{\bar \alpha}^2\psi_1+\bar \alpha\psi_2)}
	\end{align}
		\end{figure*}
	where $\bar s=s/\mu$, $\bar \alpha=\alpha/\mu$, and  $\psi_1$ and $ \psi_2 $ are given in \eqref{psi1psi}.
\end{theorem}
\begin{proof}
	The proof of Theorem \ref{cMGFtheo2} appears in  Section \ref{Non-Preemptive Policy}.
\end{proof}

\section{ The SHS Technique to Calculate MGF}\label{The SHS Technique to Calculate MGF}	
In the following, we briefly present how to use the SHS technique for our MGF analysis in Section \ref{AoI Analysis Using the SHS Technique}. We refer the readers to \cite{8469047,9103131} for more details.

The SHS technique models a queueing system  through the states $(q(t), \bold{x}(t))$, where  ${q(t)\in \mathcal{Q}=\{0,1,\ldots,m\}}$ is a continuous-time finite-state Markov
chain that  
describes the occupancy of the system and ${\bold{x}(t)=[x_0(t)\cdots x_n(t)]\in \mathbb{R}^{1\times(n+1)}}$ is a continuous
process that describes the evolution of age-related processes in the system. Following the approach in  \cite{8469047,moltafet2020sourceawareage}, we label the source of interest as source 1 and  employ the continuous process $ \bold{x}(t) $ to track the age of source 1 updates at the sink.

The Markov chain $q(t)$ can be presented as a graph $(\mathcal{Q},\mathcal{L})$ where 
each discrete state $q(t)\in \mathcal{Q}$ is a node of the chain and a (directed) link $ l\in\mathcal{L} $  from node $ q_l $ to node $q'_{l}$ indicates a transition from state $ {q_l \in \mathcal{Q}}$ to state ${q'_{l}\in \mathcal{Q}}$.

A transition occurs  when a packet arrives or departs in the system. Since the packets are generated according to the Poisson process and the service time is exponentially distributed, transition $l\in\mathcal{L}$ from state $ q_l $ to state $q'_{l}$ occurs with the  exponential rate $\lambda^{(l)}\delta_{q_l,q(t)}$,
where the Kronecker delta function $\delta_{q_l,q(t)}$ ensures that the transition $ l $ occurs only when the discrete
state $ q(t) $ is equal to $ q_l $.  When a transition $l$ occurs, the  discrete state $ q_l $ changes   to state $q'_{l}$, and the continuous state $\bold{x}$ is reset to $\bold{x}'$ according to a binary reset map matrix ${\bold{A}_l}\in\mathbb{B}^{(n+1)\times(n+1)}$ as ${\bold{x}'\!=\!\bold{x}\bold{A}_l}$. In addition, as long as  state $q(t)$ is unchanged we have  ${\dot{\bold{x}}(t)\!\triangleq\!\dfrac{\partial\bold{x}(t)}{\partial t}\!=\!\bold{1}}$, where $\bold{1}$ is the row vector $[1\cdots1]\!\in\! \mathbb{R}^{1\times(n+1)}$.

Note that unlike in a typical continuous-time Markov
chain, a transition
from a  state to itself  (i.e., a self-transition) is possible in  $q(t)\in \mathcal{Q}$. In the case of a self-transition, a reset of the continuous state $\bold{x}$ takes place, but the discrete state remains the same. In addition, for a given pair of states $\bar q,\hat q\in \mathcal{Q}$, there  may be multiple
transitions $ l $ and $ l' $ so that the discrete state changes from $ \bar q $ to
$ \hat q $ but the transition reset maps $\bold{A}_l $ and $ \bold{A}_{l'}$ are different (for more details, see  \cite[Section III]{8469047}). 

To calculate the MGF of the  AoI using the SHS technique, the state probabilities of the Markov chain,  the correlation vector between the discrete state $q(t)$ and the continuous state $\bold{x}(t)$, and  the correlation vector between the discrete state $q(t)$ and the exponential function $e^{s\bold{x}(t)},~s\in\mathbb{R}$, need to be defined. Let $\pi_q(t)$ denote the probability of being in state $q$ of the Markov chain. Let $\bold{v}_q(t)=[{v}_{q0}(t)\cdots{v}_{qn}(t)]\in\mathbb{R}^{1\times(n+1)}$  denote the correlation vector between the discrete state $q(t)$ and the continuous state $\bold{x}(t)$. Let $\bold{v}^s_q(t)=[{v}^s_{q0}(t)\cdots{v}^s_{qn}(t)]\in\mathbb{R}^{1\times(n+1)}$    denote the correlation vector between the state $q(t)$ and the exponential function $e^{s\bold{x}(t)}$. Accordingly, we have 
\begin{equation}
\pi_q(t)=\mathrm{Pr}(q(t)=q)=\mathbb{E}[\delta_{q,q(t)}], \,\,\,\forall q\in\mathcal{Q},
\end{equation}
\begin{equation}
\bold{v}_q(t)=[{v}_{q0}(t)\cdots{v}_{qn}(t)]=\mathbb{E}[\bold{x}(t)\delta_{q,q(t)}],\,\,\,\forall q\in\mathcal{Q},
\end{equation}
\begin{equation}
\bold{v}^s_q(t)=[{v}^s_{q0}(t)\cdots{v}^s_{qn}(t)]=\mathbb{E}[e^{s\bold{x}(t)}\delta_{q,q(t)}],\,\,\,\forall q\in\mathcal{Q}.
\end{equation}

Let $\mathcal{L}'_q$ denote the set of incoming transitions and  $\mathcal{L}_q$ denote the set of outgoing transitions for state $q$, defined as 
\begin{align}\nonumber
&\mathcal{L}'_q=\{l\in\mathcal{L}:q'_{l}=q\},~~~~\mathcal{L}_q=\{l\in\mathcal{L}:q_{l}=q\},\,\,\,\forall q\in\mathcal{Q}.
\end{align}
Following the ergodicity assumption of the Markov chain $q(t)$ in the AoI analysis \cite{8469047,9103131}, the state
probability vector $\boldsymbol{\pi}(t)=[\pi_0(t) \cdots \pi_m(t)]$ converges uniquely 
to the  stationary vector $\bar{\boldsymbol{\pi}}=[\bar{\pi}_0 \cdots \bar{\pi}_m]$ satisfying
\begin{align}\label{prob}
\bar{{\pi}}_q\textstyle\sum_{l\in\mathcal{L}_q}\lambda^{(l)}=\textstyle\sum_{l\in\mathcal{L}'_q}\lambda^{(l)}\bar{{\pi}}_{q_l}, \,\forall q\in\mathcal{Q},
\end{align}
\begin{align}\label{prob1}
\textstyle\sum_{q\in\mathcal{Q}}\bar{{\pi}}_q=1.
\end{align}
Further,  it has been shown in \cite[Theorem 1]{9103131} that under the ergodicity assumption of the Markov chain $q(t)$, if we can find a non-negative limit ${\bar{\bold{v}}_q=[\bar{v}_{q0} \cdots \bar{v}_{qn} ], \forall q\in \mathcal{Q}}$, for the correlation vector ${\bold{v}_q(t)}$ satisfying
\begin{align}\label{asleq}	
\bar{\bold{v}}_q\textstyle\sum_{l\in\mathcal{L}_q}\lambda^{(l)}=\bar{{\pi}}_q\bold{1}+\textstyle\sum_{l\in\mathcal{L}'_q}\lambda^{(l)}\bar{\bold{v}}_{q_l}\bold{A}_l, \,\,\,\forall q\in\mathcal{Q},
\end{align}
there exists $s_0>0$ such that for all $s<s_0$, $\bold{v}^s_q(t), \forall q\in \mathcal{Q},$ converges to  $\bold{\bar v}^s_q$ that satisfies
\begin{align}\label{aslieq}	
\bold{\bar v}^s_q\textstyle\sum_{l\in\mathcal{L}_q}\lambda^{(l)}\!=\!s\bold{\bar v}^s_q\!+\!\textstyle\sum_{l\in\mathcal{L}'_q}\lambda^{(l)}[\bold{\bar v}^s_{q_l}\bold{A}_l\!+\!\bar{{\pi}}_{q_l}\bold{1}\bold{\hat A}_l], \,\,\,\forall q\in\mathcal{Q},
\end{align}
where $\bold{\hat A}_l\in\mathbb{B}^{(n+1)\times(n+1)}$ is a binary  matrix whose ${k,j}$th element, $\bold{\hat A}_l(k,j)$, is given as 
\begin{align}\label{Ahat}
\bold{\hat A}_l(k,j)\!=\!\begin{cases}
1,\!&\!k\!=\!j, \text{and $j$th column of $ \bold{A}_l $ is a zero vector, } \\
0,\!&\!\text{otherwise.}
\end{cases}
\end{align}
Finally, the MGF of the state $\bold{x}(t)$, which is calculated by $\mathbb{E}[e^{s\bold{x}(t)}]$, converges to the stationary vector  {\cite[Theorem 1]{9103131}}
\begin{align}\label{AOIANAL}	
\mathbb{E}[e^{s\bold{x}}]=\textstyle\sum_{q\in\mathcal{Q}}\bold{\bar v}^s_q.
\end{align}
As \eqref{AOIANAL} implies, if we set the first element of continuous state $\bold{x}(t)$ to represent the AoI of source 1 at the sink, the MGF of the AoI of source 1 at the sink converges to 
\begin{align}\label{MGFage}
M_{\Delta_1}(s)=\textstyle\sum_{q\in\mathcal{Q}}{\bar v}^s_{q0}.
\end{align}
Thus,  the main challenge in calculating the MGF of the AoI of source 1 using the SHS technique reduces to deriving the first elements of correlation vectors $\bold{\bar v}^s_q,$ 
${\forall{q}\in\mathcal{Q}}$.

\section{ AoI Analysis Using the SHS Technique}\label{AoI Analysis Using the SHS Technique}

 In this section, we use  the SHS technique to calculate the MGF in \eqref{MGFage} of each source under the considered packet management policies described in Section \ref{SGP-SGP Packet Management Policy}. 
 
 \subsection{Preemptive Policy}\label{Source-Agnostic Preemptive Policy}

 The discrete states are $0=00,~1=10,~2=20,~3=01,~4=11,~5=21,~6=02,~7=12,~8=22$. State $q=a_1a_2$ indicates that  a packet of source $ a_1$  is in the transmitter server and a packet of source $ a_2 $ is in the sink server. Note that ${a_1=0}$ (resp. $ {a_2=0}$) indicates that the transmitter server (resp. the sink server) is idle.  
 
The continuous process is ${\bold{x}(t)=[x_0(t)~x_1(t)~x_2(t)]}$, where $x_0(t)$ is the current AoI of source 1 at time instant $t$, $\Delta_1(t)$, $ x_1(t)$ encodes  what $\Delta_1(t)$ would become if the packet that is in  the sink server is delivered to the sink at time instant $t$, and  $ x_2(t)$ encodes  what $\Delta_1(t)$ would become if the packet that is in  the transmitter server is delivered to the sink at time instant $t$.

Recall that  to calculate the MGF of the  AoI of source 1 in \eqref{MGFage} we need to find ${\bar v}^s_{q0}, \forall q\in\mathcal{Q},$ which are the solution of  the system of linear equations in  \eqref{aslieq}	with variables $\bold{\bar v}^s_q, \forall{q}\in\mathcal{Q}$. To form the system of linear equations \eqref{asleq} and \eqref{aslieq}, for each state ${q}\in\mathcal{Q}$, we need to determine $\bar \pi_q$ and  $\bold{A}_l$ and   $\bold{\hat A}_l$ for each incoming transition ${l\in\mathcal{L}'_q}$, carried out next.

\subsubsection{Calculation of $\bar\pi_q~\forall q\in\mathcal{Q}$}
The Markov chain for the discrete state $q(t)$  is shown in Fig. \ref{Preemptive pol}. Using \eqref{prob}, \eqref{prob1}, and the transition rates among the different states illustrated in Fig.~\ref{Preemptive pol}, it can be shown that the stationary probability vector  $\bar{\boldsymbol{\pi}}$ satisfies the following equations
 \begin{align}\label{eqlineforprob}
 &\lambda\bar{\pi}_0=\alpha\bar{\pi}_3+\alpha\bar{\pi}_6,\\&\nonumber
 (\lambda+\mu)\bar{\pi}_1=\lambda_1\bar{\pi}_0+ \lambda_1\bar{\pi}_1+\lambda_1\bar{\pi}_2+ \alpha\bar{\pi}_4+\alpha\bar{\pi}_7,\\&\nonumber
 (\lambda_1+\mu)\bar{\pi}_2=\lambda_2\bar{\pi}_0 +\lambda_2\bar{\pi}_1+\lambda_2\bar{\pi}_2+ \alpha\bar{\pi}_5+\alpha\bar{\pi}_8,\\&\nonumber
 (\lambda+\alpha)\bar{\pi}_3= \mu\bar{\pi}_1 + \mu\bar{\pi}_4+\mu\bar{\pi}_7,\\&\nonumber
 (\lambda+\mu+\alpha)\bar{\pi}_4=\lambda_1\bar{\pi}_3+\lambda_1\bar{\pi}_4+\lambda_1\bar{\pi}_5,\\&\nonumber
 (\lambda_1+\mu+\alpha)\bar{\pi}_5=\lambda_2\bar{\pi}_3+\lambda_2\bar{\pi}_4+\lambda_2\bar{\pi}_5,\\&\nonumber
 (\lambda+\alpha)\bar{\pi}_6=\mu\bar{\pi}_2+\mu\bar{\pi}_5+\mu\bar{\pi}_8,\\&\nonumber
 (\lambda+\mu+\alpha)\bar{\pi}_7=\lambda_1\bar{\pi}_6+\lambda_1\bar{\pi}_7+\lambda_1\bar{\pi}_8,\\&\nonumber
 (\lambda_1+\mu+\alpha)\bar{\pi}_8=\lambda_2\bar{\pi}_6+\lambda_2\bar{\pi}_7+\lambda_2\bar{\pi}_8,\\&\nonumber
 \textstyle\sum_{q\in\mathcal{Q}}\bar{{\pi}}_q=1.
 \end{align}
 Solving \eqref{eqlineforprob}, the stationary probabilities are given as
 \begin{align}\nonumber
 &\bar{\pi}_0=\! \dfrac{\alpha\mu}{\bar\Psi},~
 \bar{\pi}_1\!=\!\dfrac{\alpha\lambda_1(\alpha+\mu+\lambda)}{(\mu+\alpha)\bar\Psi},~
 \bar{\pi}_2\!=\!\dfrac{\alpha\lambda_2(\alpha+\mu+\lambda)}{(\mu+\alpha)\bar\Psi},\\&\label{proeqq1}
 \bar{\pi}_3\!=\!\dfrac{\lambda_1\mu}{\bar\Psi},~
 \bar{\pi}_4\!=\!\dfrac{\lambda_1^2\mu}{(\mu+\alpha)\bar\Psi},~~~~~~~~
 \bar{\pi}_5\!=\!\dfrac{\lambda_1\lambda_2\mu}{(\mu+\alpha)\bar\Psi},\\&\nonumber
 \bar{\pi}_6=\dfrac{\lambda_2\mu}{\bar\Psi},~
 \bar{\pi}_7\!=\!\dfrac{\lambda_1\lambda_2\mu}{(\mu+\alpha)\bar\Psi},~~~~~~~
 \bar{\pi}_8\!=\!\dfrac{\lambda_2^2\mu}{(\mu+\alpha)\bar\Psi},
 \end{align}
 where $\bar\Psi=(\lambda+\alpha)(\lambda+\mu)$.

 \begin{figure}
 	\centering
 	\includegraphics[width=0.85\linewidth,trim = 0mm 0mm 0mm 0mm,clip]{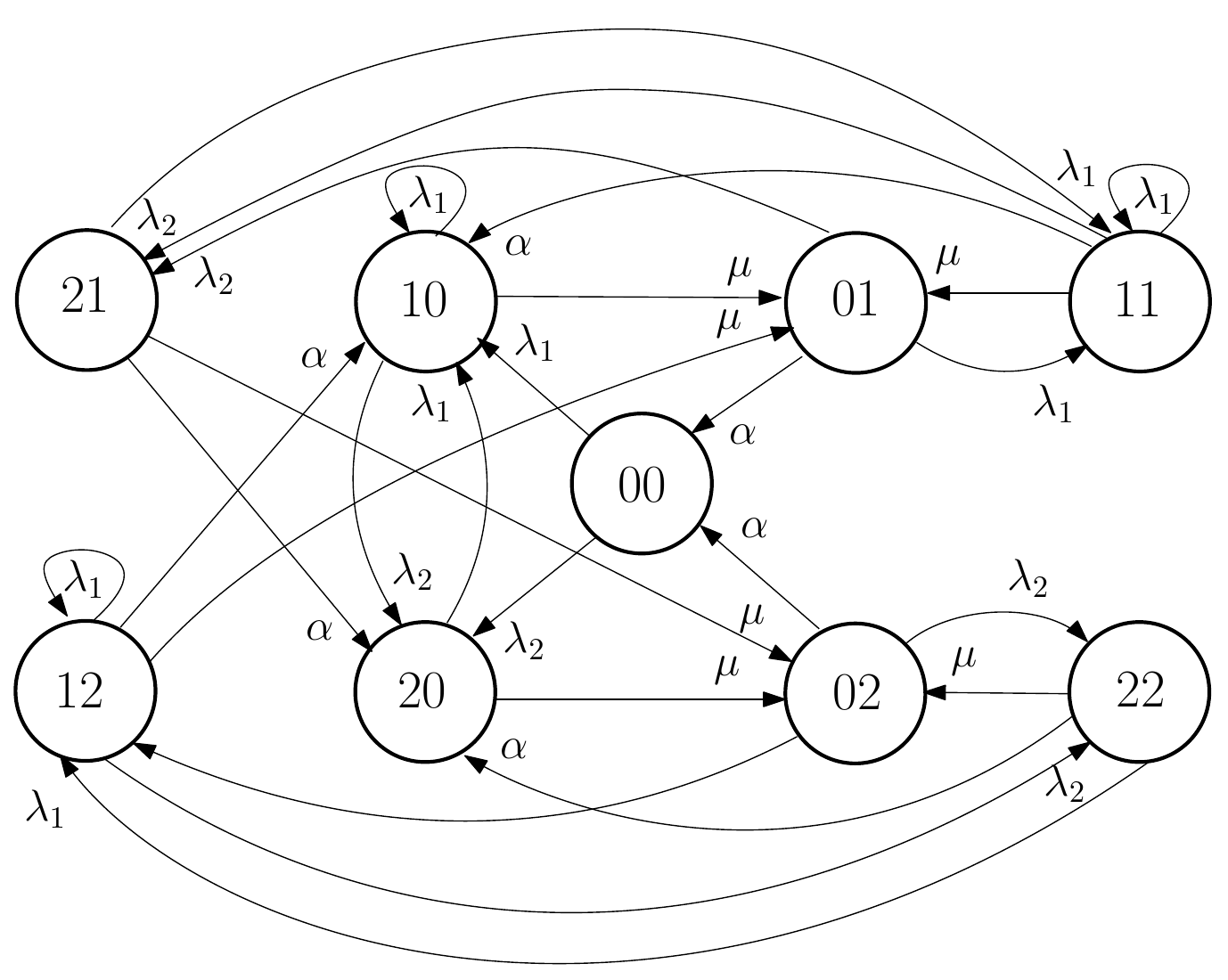}
 	\caption{The SHS Markov chain for the preemptive policy. }  
 	\label{Preemptive pol}
 \end{figure}

\subsubsection{Transition reset map matrices $\bold{A}_l$ and   $\bold{\hat A}_l$}
 The transitions and their effects on the continuous state $\bold{x}(t)$ are summarized in Table \ref{Preemptive T2}.   In the following, we explain the transitions in detail.
 \begin{table}
 	\centering\small
 	\caption{Table of transitions for the Markov chain in Fig. \ref{Preemptive pol}}
 	\label{Preemptive T2}
 	\begin{tabular}{ |l|l|c|c|c|c|}
 		\hline
 		\textit{l}  & $q_l \rightarrow q'_l $&$\lambda^{(l)}$& $\bold{x}\bold{A}_l$&$\bold{A}_l$&$\bold{\hat A}_l$\\			
 		\hline			
 		1&$00 \rightarrow  10$& $\lambda_1$&$\left[x_0 ~ x_1 ~ 0 \right]$&$\tiny\begin{bmatrix}
 		1 & 0&0 \\
 		0 & 1&0 \\
 		0 & 0&0 
 		\end{bmatrix}$&$\tiny\begin{bmatrix}
 		0 & 0&0 \\
 		0 & 0&0 \\
 		0 & 0&1 
 		\end{bmatrix}$
 		\\		
 		\hline
 		2&$00 \rightarrow  20$& $\lambda_2$&$\left[x_0 ~ x_1 ~ x_1 \right]$&$\tiny\begin{bmatrix}
 		1 & 0&0 \\
 		0 & 1&1 \\
 		0 & 0&0 
 		\end{bmatrix}$&$\tiny\begin{bmatrix}
 		0 & 0&0 \\
 		0 & 0&0 \\
 		0 & 0&0 
 		\end{bmatrix}$
 		\\		
 		\hline
 		3&$10 \rightarrow  10$& $\lambda_1$&$\left[x_0 ~ x_1 ~0 \right]$&$\tiny\begin{bmatrix}
 		1 & 0&0 \\
 		0 & 1&0 \\
 		0 & 0&0 
 		\end{bmatrix}$&$\tiny\begin{bmatrix}
 		0 & 0&0 \\
 		0 & 0&0 \\
 		0 & 0&1 
 		\end{bmatrix}$
 		\\		
 		\hline
 		4&$10 \rightarrow  20$& $\lambda_2$&$\left[x_0 ~ x_1~x_1 \right]$&$\tiny\begin{bmatrix}
 		1 & 0&0 \\
 		0 & 1&1 \\
 		0 & 0&0 
 		\end{bmatrix}$&$\tiny\begin{bmatrix}
 		0 & 0&0 \\
 		0 & 0&0 \\
 		0 & 0&0 
 		\end{bmatrix}$
 		\\		
 		\hline
 		5&$10 \rightarrow  01$& $\mu$&$\left[x_0 ~ x_2 ~ x_2 \right]$&$\tiny\begin{bmatrix}
 		1 & 0&0 \\
 		0 & 0&0 \\
 		0 & 1&1 
 		\end{bmatrix}$&$\tiny\begin{bmatrix}
 		0 & 0&0 \\
 		0 & 0&0 \\
 		0 & 0&0 
 		\end{bmatrix}$
 		\\		
 		\hline
 		6&$20 \rightarrow  10$& $\lambda_1$&$\left[x_0 ~ x_1 ~ 0\right]$&$\tiny\begin{bmatrix}
 		1 & 0&0 \\
 		0 & 1&0 \\
 		0 & 0&0 
 		\end{bmatrix}$&$\tiny\begin{bmatrix}
 		0 & 0&0 \\
 		0 & 0&0 \\
 		0 & 0&1 
 		\end{bmatrix}$
 		\\		
 		\hline
 		7&$20 \rightarrow  02$& $\mu$&$\left[x_0 ~ x_0 ~ x_2 \right]$&$\tiny\begin{bmatrix}
 		1 & 1&0 \\
 		0 & 0&0 \\
 		0 & 0&1 
 		\end{bmatrix}$&$\tiny\begin{bmatrix}
 		0 & 0&0 \\
 		0 & 0&0 \\
 		0 & 0&0 
 		\end{bmatrix}$
 		\\		
 		\hline
 		8&$01 \rightarrow  11$& $\lambda_1$&$\left[x_0 ~ x_1 ~0 \right]$&$\tiny\begin{bmatrix}
 		1 & 0&0 \\
 		0 & 1&0 \\
 		0 & 0&0 
 		\end{bmatrix}$&$\tiny\begin{bmatrix}
 		0 & 0&0 \\
 		0 & 0&0 \\
 		0 & 0&1 
 		\end{bmatrix}$
 		\\		
 		\hline
 		9&$01 \rightarrow  21$& $\lambda_2$&$\left[x_0 ~ x_1 ~x_1  \right]$&$\tiny\begin{bmatrix}
 		1 & 0&0 \\
 		0 & 1&1 \\
 		0 & 0&0 
 		\end{bmatrix}$&$\tiny\begin{bmatrix}
 		0 & 0&0 \\
 		0 & 0&0 \\
 		0 & 0&0 
 		\end{bmatrix}$
 		\\		
 		\hline
 		10&$01 \rightarrow  00$& $\alpha$&$\left[x_1 ~ x_1 ~ x_2 \right]$&$\tiny\begin{bmatrix}
 		1 & 1&0 \\
 		0 & 0&0 \\
 		0 & 0&1 
 		\end{bmatrix}$&$\tiny\begin{bmatrix}
 		0 & 0&0 \\
 		0 & 0&0 \\
 		0 & 0&0 
 		\end{bmatrix}$
 		\\		
 		\hline
 		11&$11 \rightarrow  11$& $\lambda_1$&$\left[x_0 ~ x_1 ~ 0  \right]$&$\tiny\begin{bmatrix}
 		1 & 0&0 \\
 		0 & 1&0 \\
 		0 & 0&0 
 		\end{bmatrix}$&$\tiny\begin{bmatrix}
 		0 & 0&0 \\
 		0 & 0&0 \\
 		0 & 0&1 
 		\end{bmatrix}$
 		\\		
 		\hline
 		12&$11 \rightarrow  21$& $\lambda_2$&$\left[x_0 ~ x_1 ~ x_1  \right]$&$\tiny\begin{bmatrix}
 		1 & 0&0 \\
 		0 & 1&1 \\
 		0 & 0&0 
 		\end{bmatrix}$&$\tiny\begin{bmatrix}
 		0 & 0&0 \\
 		0 & 0&0 \\
 		0 & 0&0 
 		\end{bmatrix}$
 		\\		
 		\hline
 		13&$11 \rightarrow  01$& $\mu$&$\left[x_0 ~ x_2 ~ x_2  \right]$&$\tiny\begin{bmatrix}
 		1 & 0&0 \\
 		0 & 0&0 \\
 		0 & 1&1 
 		\end{bmatrix}$&$\tiny\begin{bmatrix}
 		0 & 0&0 \\
 		0 & 0&0 \\
 		0 & 0&0 
 		\end{bmatrix}$
 		\\		
 		\hline
 		14&$11 \rightarrow  10$& $\alpha$&$\left[x_1 ~ x_1 ~x_2  \right]$&$\tiny\begin{bmatrix}
 		0 & 0&0 \\
 		1 & 1&0 \\
 		0 & 0&1 
 		\end{bmatrix}$&$\tiny\begin{bmatrix}
 		0 & 0&0 \\
 		0 & 0&0 \\
 		0 & 0&0 
 		\end{bmatrix}$
 		\\		
 		\hline
 		15&$21 \rightarrow  11$& $\lambda_1$&$\left[x_0 ~ x_1 ~ 0 \right]$&$\tiny\begin{bmatrix}
 		1 & 0&0 \\
 		0 & 1&0 \\
 		0 & 0&0 
 		\end{bmatrix}$&$\tiny\begin{bmatrix}
 		0 & 0&0 \\
 		0 & 0&0 \\
 		0 & 0&1 
 		\end{bmatrix}$
 		\\		
 		\hline
 		16&$21 \rightarrow  02$& $\mu$&$\left[x_0 ~ x_0 ~ x_2  \right]$&$\tiny\begin{bmatrix}
 		1 & 1&0 \\
 		0 & 0&0 \\
 		0 & 0&1 
 		\end{bmatrix}$&$\tiny\begin{bmatrix}
 		0 & 0&0 \\
 		0 & 0&0 \\
 		0 & 0&0 
 		\end{bmatrix}$
 		\\		
 		\hline
 		17&$21 \rightarrow  20$& $\alpha$&$\left[x_1 ~ x_1 ~ x_1 \right]$&$\tiny\begin{bmatrix}
 		0 & 0&0 \\
 		1 & 1&1 \\
 		0 & 0&0 
 		\end{bmatrix}$&$\tiny\begin{bmatrix}
 		0 & 0&0 \\
 		0 & 0&0 \\
 		0 & 0&0 
 		\end{bmatrix}$
 		\\		
 		\hline
 		18&$02 \rightarrow  12$& $\lambda_1$&$\left[x_0 ~ x_0 ~0 \right]$&$\tiny\begin{bmatrix}
 		1 & 1&0 \\
 		0 & 0&0 \\
 		0 & 0&0 
 		\end{bmatrix}$&$\tiny\begin{bmatrix}
 		0 & 0&0 \\
 		0 & 0&0 \\
 		0 & 0&1 
 		\end{bmatrix}$
 		\\		
 		\hline
 		19&$02 \rightarrow  22$& $\lambda_2$&$\left[x_0 ~ x_0 ~ x_0  \right]$&$\tiny\begin{bmatrix}
 		1 & 1&1 \\
 		0 & 0&0 \\
 		0 & 0&0 
 		\end{bmatrix}$&$\tiny\begin{bmatrix}
 		0 & 0&0 \\
 		0 & 0&0 \\
 		0 & 0&0 
 		\end{bmatrix}$
 		\\		
 		\hline
 		20&$02 \rightarrow  00$& $\alpha$&$\left[x_0 ~ x_1 ~ x_2 \right]$&$\tiny\begin{bmatrix}
 		1 & 0&0 \\
 		0 & 1&0 \\
 		0 & 0&1 
 		\end{bmatrix}$&$\tiny\begin{bmatrix}
 		0 & 0&0 \\
 		0 & 0&0 \\
 		0 & 0&0 
 		\end{bmatrix}$
 		\\		
 		\hline
 		21&$12 \rightarrow  12$& $\lambda_1$&$\left[x_0 ~ x_0 ~ 0  \right]$&$\tiny\begin{bmatrix}
 		1 & 1&0 \\
 		0 & 0&0 \\
 		0 & 0&0 
 		\end{bmatrix}$&$\tiny\begin{bmatrix}
 		0 & 0&0 \\
 		0 & 0&0 \\
 		0 & 0&1 
 		\end{bmatrix}$
 		\\		
 		\hline
 		22&$12 \rightarrow  22$& $\lambda_2$&$\left[x_0 ~ x_0 ~ x_0  \right]$&$\tiny\begin{bmatrix}
 		1 & 1&1 \\
 		0 & 0&0 \\
 		0 & 0&0 
 		\end{bmatrix}$&$\tiny\begin{bmatrix}
 		0 & 0&0 \\
 		0 & 0&0 \\
 		0 & 0&0 
 		\end{bmatrix}$
 		\\		
 		\hline
 		23&$12 \rightarrow  01$& $\mu$&$\left[x_0 ~ x_2 ~ x_2 \right]$&$\tiny\begin{bmatrix}
 		1 & 0&0 \\
 		0 & 0&0 \\
 		0 & 1&1 
 		\end{bmatrix}$&$\tiny\begin{bmatrix}
 		0 & 0&0 \\
 		0 & 0&0 \\
 		0 & 0&0 
 		\end{bmatrix}$
 		\\		
 		\hline
 		24&$12 \rightarrow  10$& $\alpha$&$\left[x_0 ~ x_1 ~ x_2  \right]$&$\tiny\begin{bmatrix}
 		1 & 0&0 \\
 		0 & 1&0 \\
 		0 & 0&1 
 		\end{bmatrix}$&$\tiny\begin{bmatrix}
 		0 & 0&0 \\
 		0 & 0&0 \\
 		0 & 0&0 
 		\end{bmatrix}$	\\		
 		\hline
 			25&$22 \rightarrow  12$& $\lambda_1$&$\left[x_0 ~ x_0 ~ 0  \right]$&$\tiny\begin{bmatrix}
 		1 & 1&0 \\
 		0 & 0&0 \\
 		0 & 0&0 
 		\end{bmatrix}$&$\tiny\begin{bmatrix}
 		0 & 0&0 \\
 		0 & 0&0 \\
 		0 & 0&1 
 		\end{bmatrix}$	\\		
 		\hline
 			26&$22 \rightarrow  02$& $\mu$&$\left[x_0 ~ x_0 ~ x_2  \right]$&$\tiny\begin{bmatrix}
 		1 & 1&0 \\
 		0 & 0&0 \\
 		0 & 0&1 
 		\end{bmatrix}$&$\tiny\begin{bmatrix}
 		0 & 0&0 \\
 		0 & 0&0 \\
 		0 & 0&0 
 		\end{bmatrix}$	\\		
 		\hline
 			27&$22 \rightarrow  20$& $\alpha$&$\left[x_0 ~ x_1 ~ x_1  \right]$&$\tiny\begin{bmatrix}
 		1 & 0&0 \\
 		0 & 1&1 \\
 		0 & 0&0 
 		\end{bmatrix}$&$\tiny\begin{bmatrix}
 		0 & 0&0 \\
 		0 & 0&0 \\
 		0 & 0&0 
 		\end{bmatrix}$	\\		
 		\hline
 	\end{tabular}
 \end{table}

 \begin{itemize}
 	\item \textit{l=1}: A source 1 packet  arrives at  the transmitter server. With this transition the AoI of source 1 does not change, i.e., $x'_0=x_0$. This is because the arrival of source 1 packet does not yield an age reduction until it is delivered to the sink. Since the sink server is
 	empty, $x_1$ becomes irrelevant to the AoI of source 1, and thus, it does not matter
 	what is assigned to $x'_1$. In this paper, when the system moves into a new state where $x_i$ is
 	irrelevant, we set $x'_i=x_i,~i\in\{1,2\}$. Thus, for the transition $l=1$, we have $x_1'=x_1$. Since the arriving source 1 packet is fresh and its age is zero, we have $x'_2=0$.  
 	Finally, we have 
 	\begin{align}\label{l1}
 	\bold{x}'=[x_0~x_1~x_2]\bold{A}_1=[x_0 ~ x_1 ~ 0 ].
 	\end{align} 
 	According to \eqref{l1}, it can be shown that the binary transition reset map matrix $\bold{A}_1$ is given by
 	$
 	\small	\bold{A}_1=\begin{bmatrix}
 	1 & 0&0 \\
 	0 & 1&0 \\
 	0 & 0&0 
 	\end{bmatrix}.
 	$
 	Then, by using \eqref{Ahat}, $\bold{\hat A}_1$ is given as
 	$
 	\small	\bold{\hat A}_1=\begin{bmatrix}
 	0 & 0&0 \\
 	0 & 0&0 \\
 	0 & 0&1 
 	\end{bmatrix}.
 	$
 	As it can be seen, when we have $\bold{x}'$ for a transition $l\in\mathcal{L}$, it is easy to calculate $\bold{A}_l$ and  $\bold{\hat A}_l$.
 	Thus, for the rest of the transitions, we just explain the calculation of $\bold{x}'$ and present the final expressions of $\bold{A}_l$ and  $\bold{\hat A}_l$.

 	\item \textit{l=2}: A source 2 packet arrives at the transmitter server. Similarly as for transition $l=1$, we have  $x'_0=x_0$ and $x'_1=x_1$.  Since the arriving packet is a source 2
 	packet, its delivery does not change the AoI of source 1,
 	and thus, we have  $x'_2=x_1$.  
 	
 	\item \textit{l=3}: A source 1 packet is in the  transmitter server, the sink server is idle, and a source 1 packet arrives. According to the preemptive policy, the source 1 packet that is under service is preempted by the arriving packet. Similarly as for transition $l=1$, we have  $x'_0=x_0$ and $x'_1=x_1$. Since the arriving source 1 packet is fresh and its age is zero, we have $x'_2=0$. The reset maps of transitions $l=11$ and  $l=21$  can be derived similarly.
 	
 	\item \textit{l=4}:
 	 A source 1 packet is in the  transmitter server, the sink server is idle, and a source 2 packet arrives. According to the preemptive policy, the source 1 packet that is under service is preempted by the arriving packet. Similarly as for transition $l=1$, we have  $x'_0=x_0$ and $x'_1=x_1$. Since the arriving packet is a source 2 packet, its delivery does not change the AoI of source 1,
 	 and thus we have  $x'_2=x_1$. The reset maps of transitions $l=12$ and $l=22$ can be derived similarly.
 	
 	\item \textit{l=5}:	The sink server is idle, and the source 1 packet in the transmitter server completes service and moves to the sink server. With this transition, we have $x'_0=x_0$ because this transition does not change the AoI at the sink. Since the source 1 packet moves to the sink server and its delivery would reduce the AoI to $x_2$, we have $x'_1=x_2$. Since the transmitter server is	empty, $x_2$ becomes irrelevant, and thus, we have $x'_2=x_2$. The reset maps of transitions $l=13$ and $l=23$ can be derived similarly.
 	
 	\item \textit{l=6}: A source 2 packet is in the  transmitter server, the sink server is idle, and a source 1 packet arrives. According to the preemptive policy,  the source 2 packet that is under service is preempted by the arriving packet. Similarly as for transition $l=1$, we have $x'_0=x_0$ and $x'_1=x_1$. Since the arriving source 1 packet is fresh and its age is zero, we have $x'_2=0$. The reset maps of transitions $l=15$  and  $l=25$ can be derived similarly.
 	
 	\item \textit{l=7}:	The sink server is idle and the source 2 packet in the transmitter server completes service and moves to the sink server. With this transition, we have $x'_0=x_0$, and since the source 2 packet delivery does not change the AoI of source 1, we have $x'_1=x_0$. Since the transmitter server is empty, $x_2$ becomes irrelevant, and thus, we have $x'_2=x_2$. The reset maps of transitions $l=16$ and $l=26$ can be derived similarly.
 	
 	\item \textit{l=8}: The transmitter server is idle, a source 1 packet is in the sink server, and a source 1 packet arrives. With this transition, we have $x'_0=x_0$. The delivery of the packet in the sink server would reduce the AoI to $x_1$ thus, we have $x'_1=x_1$. Since the arriving source 1 packet is fresh and its age is zero, we have $x'_2=0$.	The reset maps of transition $ l = 18 $ can be	derived similarly.

 	\item \textit{l=9}: The transmitter server is idle, a source 1 packet is in the sink server, and a source 2 packet arrives. Similarly as for transition $l=8$, we have  $x'_0=x_0$ and $x'_1=x_1$. Since the arriving packet is a source 2 packet, its delivery does not change the AoI of source 1, and thus we have $x'_2=x_1$.	The reset maps of transition $ l = 19 $ can be	derived similarly.

 	\item \textit{l=10}: The transmitter server is idle and the source 1 packet in the sink server completes service and is delivered to the sink. With this transition, we have ${x'_0=x_1}$. Since  the servers become idle we have ${x'_1=x_1}$ and ${x'_2=x_2}$. The reset maps of transitions
 	$l\in\{14,17,20,24,27\}$ can be derived similarly.

 \end{itemize}

 \subsubsection{ Calculation of the MGF of the AoI}
 Recall from Section \ref{The SHS Technique to Calculate MGF} that to calculate MGF of the AoI, we need to first ensure whether we can find  non-negative vectors ${\bar{\bold{v}}_q=[\bar{v}_{q0} \cdots \bar{v}_{qn} ], \forall q\in \mathcal{Q}},$ that satisfy \eqref{asleq}. Having derived  the matrices $\bold{A}_l$ for all transitions (see Table~\ref{Preemptive T2}), we can form the system of linear equations in  \eqref{asleq}. 
 It can be shown that the system of linear equations in \eqref{asleq} has a non-negative solution. Consequently, using  the derived $\bold{A}_l$ and $\bold{\hat A}_l$ for the transitions in Table~\ref{Preemptive T2}, we can form the system of linear equations in  \eqref{aslieq}. Finally, by solving the formed system of linear equations,  the values of $\bar{v}^s_{q0},~ \forall q\in\mathcal{Q}$, are
 calculated as presented in Appendix~\ref{Valuesof  mgf preemptive}. Substituting  ${\bar v}^s_{q0}, \forall q\in\mathcal{Q}$ into \eqref{MGFage} results  the MGF of the AoI of source 1 under the preemptive policy, given in Theorem~\ref{cMGFtheo1}.

\subsection{Blocking Policy}\label{Non-Preemptive Policy}

The elements of discrete state space $ \mathcal{Q} $ are  $0=\text{II},~1=\text{IB},~2=\text{BI},~3=\text{BB}$, where ``I'' stands for ``idle'' and ``B'' stands for ``busy''. State $0=\text{II}$ indicates that both the transmitter and sink servers are idle; $1=\text{IB}$ indicates that the transmitter server is idle and there is a packet (regardless of the source) under service in the sink server; $2=\text{BI}$ indicates that there is a packet under service in the transmitter server and the sink server is idle; and  $3=\text{BB}$ indicates that both the transmitter and sink servers are busy. The  continuous state space is  the same as that for the  preemptive policy. 

Next, we determine $\bar \pi_q, ~\forall q\in\mathcal{Q}$, $\bold{A}_l$, and   $\bold{\hat A}_l$ for each incoming transition ${l\in\mathcal{L}'_q}$ which are needed to form the system of linear equations  \eqref{asleq} and \eqref{aslieq}.

\subsubsection{Calculation of $\bar\pi_q~\forall q\in\mathcal{Q}$}
The Markov chain for the discrete state $q(t)$  is shown in Fig. \ref{non-preemptive makov}. Using \eqref{prob} and the transition rates among the different states illustrated in Fig.~\ref{non-preemptive makov}, it can be shown that the stationary probability vector  $\bar{\boldsymbol{\pi}}$ satisfies $\bar{\boldsymbol{\pi}}\bold{D}=\bar{\boldsymbol{\pi}}\bold{Q}$ with $\bold{D}=\text{diag}[\lambda,\lambda_1+\mu,\lambda+\alpha,\lambda_1+\mu+\alpha]$ and
\begin{align}\nonumber
\bold{Q}=\left[\begin{array}{cccc}
0      & \lambda   & 0  &0 \\
0      & \lambda_1 & \mu&0 \\
\alpha & 0         & 0  &\lambda \\
0      & \alpha    & \mu&\lambda_1 \\
\end{array}\right]\!.
\end{align}
Using the above $
\bar{\boldsymbol{\pi}}\bold{D}=\bar{\boldsymbol{\pi}}\bold{Q} 
$  and  $\textstyle\sum_{q\in\mathcal{Q}}\bar{{\pi}}_q=1$ in \eqref{prob1}, the stationary probabilities are given as
\begin{align}\label{proeqq0}
\bar{\boldsymbol{\pi}}\!=\!\dfrac{1}{(\lambda\!+\!\mu)(\lambda\!+\!\alpha)}\left[\alpha\mu, \dfrac{\alpha\lambda(\alpha\!+\!\lambda\!+\!\mu)}{\alpha+\mu}, \mu\lambda,\dfrac{\mu\lambda^2}{\alpha\!+\!\mu}\right].
\end{align}

\subsubsection{Transition reset map matrices $\bold{A}_l$ and   $\bold{\hat A}_l$}
The transitions, their effects on the continuous state $\bold{x}(t)$, and the values of $\bold{A}_l$ and   $\bold{\hat A}_l$  are summarized in Table~\ref{non-preemptive}. Following the arguments of Section~\ref{Source-Agnostic Preemptive Policy}, we simply present the final values of  $\bold{A}_l$ and   $\bold{\hat A}_l$   in Table~\ref{non-preemptive}. 

\begin{figure}
	\centering
	\includegraphics[width=0.4\linewidth,trim = 0mm 0mm 0mm 0mm,clip]{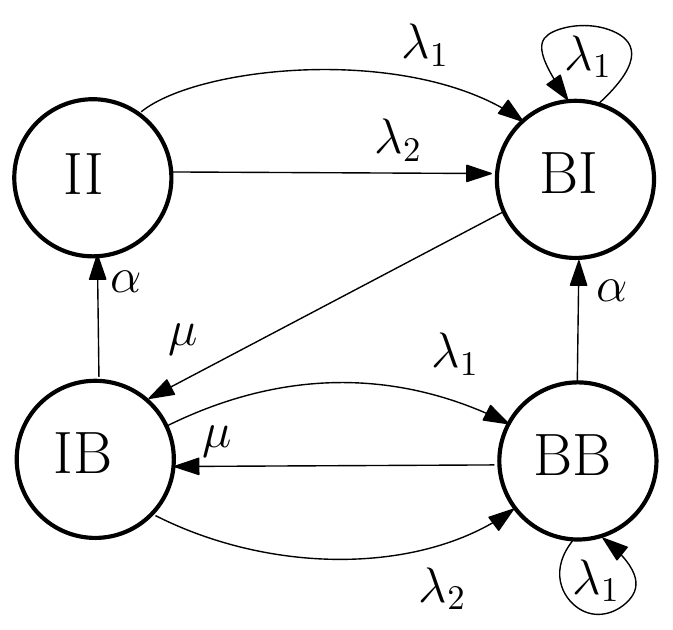}
	\caption{The SHS Markov chain for the blocking policy. }  
	\label{non-preemptive makov}
\end{figure}

\begin{table}
	\centering\small
	\caption{Table of transitions for the Markov chain in Fig. \ref{non-preemptive makov}}
	\label{non-preemptive}
	\begin{tabular}{ |l|l|c|c|c|c|}
		\hline
		\textit{l}  & $q_l \rightarrow q'_l $&$\lambda^{(l)}$& $\bold{x}\bold{A}_l$&$\bold{A}_l$& $\bold{\hat A}_l$ \\			
		\hline			
		1&$\text{II} \rightarrow  \text{BI}\!\!$& $\lambda_1\!\!$&$\left[x_0 ~ x_1 ~ 0 \right]\!\!$&$\tiny\begin{bmatrix}
		1 & 0&0 \\
		0 & 1&0 \\
		0 & 0&0 
		\end{bmatrix}$&$\tiny\begin{bmatrix}
		0 & 0&0 \\
		0 & 0&0 \\
		0 & 0&1 
		\end{bmatrix}$
		\\		
		\hline
		2&$\text{II}\rightarrow \text{BI}\!\!$& $\lambda_2\!\!$&$\left[x_0 ~ x_1 ~ x_1 \right]\!\!$&$\tiny\begin{bmatrix}
		1 & 0&0 \\
		0 & 1&1 \\
		0 & 0&0 
		\end{bmatrix}$&$\tiny\begin{bmatrix}
		0 & 0&0 \\
		0 & 0&0 \\
		0 & 0&0 
		\end{bmatrix}$
		\\		
		\hline
		3&$\text{BI} \rightarrow  \text{BI}\!\!$& $\lambda_1\!\!$&$\left[x_0 ~ x_1 ~x_2 \right]\!\!$&$\tiny\begin{bmatrix}
		1 & 0&0 \\
		0 & 1&0 \\
		0 & 0&1 
		\end{bmatrix}$&$\tiny\begin{bmatrix}
		0 & 0&0 \\
		0 & 0&0 \\
		0 & 0&0 
		\end{bmatrix}$
		\\		
		\hline
		4&$\text{BI}\rightarrow  \text{IB}\!\!$& $\mu\!\!$&$\left[x_0 ~ x_2 ~ x_2  \right]\!\!$&$\tiny\begin{bmatrix}
		1 & 0&0 \\
		0 & 0&0 \\
		0 & 1&1
		\end{bmatrix}$&$\tiny\begin{bmatrix}
		0 & 0&0 \\
		0 & 0&0 \\
		0 & 0&0 
		\end{bmatrix}$
		\\		
		\hline
		5&$\text{IB} \rightarrow  \text{II}\!\!$& $\alpha\!\!$&$\left[x_1 ~ x_1 ~ x_2 \right]\!\!$&$\tiny\begin{bmatrix}
		0 & 0&0 \\
		1 & 1&0 \\
		0 & 0&1 
		\end{bmatrix}$&$\tiny\begin{bmatrix}
		0 & 0&0 \\
		0 & 0&0 \\
		0 & 0&0 
		\end{bmatrix}$
		\\		
		\hline
		6&$\text{IB} \rightarrow  \text{BB}\!\!$& $\lambda_1\!\!$&$\left[x_0 ~ x_1 ~ 0 \right]\!\!$&$\tiny\begin{bmatrix}
		1 & 0&0 \\
		0 & 1&0 \\
		0 & 0&0 
		\end{bmatrix}$&$\tiny\begin{bmatrix}
		0 & 0&0 \\
		0 & 0&0 \\
		0 & 0&1 
		\end{bmatrix}$
		\\		
		\hline
		7&$\text{IB} \rightarrow  \text{BB}\!\!$& $\lambda_2\!\!$&$\left[x_0 ~ x_1 ~ x_1 \right]\!\!$&$\tiny\begin{bmatrix}
		1 & 0&0 \\
		0 & 1&1 \\
		0 & 0&0 
		\end{bmatrix}$&$\tiny\begin{bmatrix}
		0 & 0&0 \\
		0 & 0&0 \\
		0 & 0&0 
		\end{bmatrix}$
		\\		
		\hline
		8&$\text{BB} \rightarrow  \text{BB}\!\!$& $\lambda_1\!\!$&$\left[x_0 ~ x_1 ~x_2\right]\!\!$&$\tiny\begin{bmatrix}
		1 & 0&0 \\
		0 & 1&0 \\
		0 & 0&1 
		\end{bmatrix}$&$\tiny\begin{bmatrix}
		0 & 0&0 \\
		0 & 0&0 \\
		0 & 0&0 
		\end{bmatrix}$
		\\		
		\hline
		9&$\text{BB} \rightarrow  \text{IB}\!\!$& $\mu\!\!$&$\left[x_0 ~ x_1 ~x_2 \right]\!\!$&$\tiny\begin{bmatrix}
		1 & 0&0 \\
		0 & 0&0 \\
		0 & 1&1
		\end{bmatrix}$&$\tiny\begin{bmatrix}
		0 & 0&0 \\
		0 & 0&0 \\
		0 & 0&0 
		\end{bmatrix}$
		\\		
		\hline
		10&$\text{BB} \rightarrow  \text{BI}\!\!$& $\alpha\!\!$&$\left[x_1 ~ x_1 ~ x_2 \right]\!\!$&$\tiny\begin{bmatrix}
		0 & 0&0 \\
		1 & 1&0 \\
		0 & 0&1 
		\end{bmatrix}$&$\tiny\begin{bmatrix}
		0 & 0&0 \\
		0 & 0&0 \\
		0 & 0&0 
		\end{bmatrix}$
		\\		
		\hline
	\end{tabular}
		\vspace{-5mm}
\end{table}

\begin{itemize}
	\item \textit{l=1}: A source 1 packet arrives at the transmitter server. We have  $x'_0=x_0$, because there is no departure at the sink server. Since the sink server is
	empty, $x_1$ becomes irrelevant to the AoI of source 1, and thus, we have $x_1'=x_1$. Since the arriving source 1 packet is fresh and its age is zero, we have $x'_2=0$.  The reset maps of transition $ l = 6 $ can be
	derived similarly.

	\item \textit{l=2}: A source 2 packet arrives at the transmitter server. Similarly as for transition $l=1$, we have  $x'_0=x_0$ and $x'_1=x_1$.  Since the arriving packet is a source 2
	packet, its delivery does not change the AoI of source 1,
	and thus we have  $x'_2=x_1$.  The reset maps of transition $ l = 7 $ can be
	derived similarly.

	\item \textit{l=3}: A packet is in the transmitter server, the sink server is idle, and a source 1 packet arrives. According to the blocking policy the arriving  packet is blocked and cleared. Similarly as for transition $l=1$, we have  $x'_0=x_0$ and $x'_1=x_1$. The delivery of the packet in the transmitter server to the sink  would reduce the AoI to $x_2$ and thus, $x'_2=x_2$. The reset maps of transition $ l = 8 $ can be	derived similarly.

	\item \textit{l=4}:	A packet in the transmitter server completes service and moves to the sink server.
		In this transition, we have $x'_0=x_0$ because this transition does not change the AoI at the sink. Since the packet moves to the sink server, we have $x'_1=x_2$. Since the transmitter server is
		empty, $x_2$ becomes irrelevant, and thus, we have $x'_2=x_2$.  The reset maps of transition $ l = 9 $ can be
		derived similarly.

		\item \textit{l=5}: A packet in the sink server completes service and is delivered to the sink.  
		With this transition, we have $x'_0=x_1$. Since with this transition the both servers become idle we have $x'_1=x_1$ and $x'_2=x_2$.  The reset maps of transition $ l = 10 $ can be
		derived similarly.
\end{itemize}

\subsubsection{ Calculation of the MGF of the AoI}
Having derived  the matrices $\bold{A}_l$ for all transitions (see Table~\ref{non-preemptive}), we can form the system of linear equations in  \eqref{asleq}. 
It can be shown that the system of linear equations in \eqref{asleq} has a non-negative solution. Using the derived $\bold{A}_l$ and $\bold{\hat A}_l$ for the transitions in Table~\ref{non-preemptive} we can form the system of linear equations in  \eqref{aslieq}.  By solving the  system of linear equations, values of ${\bar v}^s_{q0}, \forall q\in\mathcal{Q}$, are given as
\begin{align}\label{vsnonpreemptive}
&\bar {v}^s_{00}=\dfrac{-\bar \alpha^2\rho_1(\bar \alpha+\rho+1-\bar s)}{(\bar \alpha+\rho)(\rho+1)\Psi},\\&\nonumber
 \bar {v}^s_{10}=\dfrac{-\bar \alpha^2\rho_1\rho(\bar \alpha+\rho+1-\bar s)}{(\bar \alpha+1-\bar s)(1-\bar s)(\bar \alpha+\rho)(\rho+1)\Psi},\\&\nonumber
 \bar {v}^s_{20}=\dfrac{-\bar \alpha^2\rho_1\rho(\bar \alpha+\rho+1-\bar s)}{(\bar \alpha-\bar s)(1-\bar s)(\bar \alpha+\rho)(\rho+1)\Psi},\\&\nonumber
 \bar {v}^s_{30}=\dfrac{-\bar \alpha^2\rho_1\rho^2(\bar \alpha+\rho+1-\bar s)}{(\bar \alpha-\bar s)(\bar \alpha+1-\bar s)(1-\bar s)(\bar \alpha+\rho)(\rho+1)\Psi},
\end{align}
where $\bar s=s/\mu$, $\bar \alpha=\alpha/\mu$,  and $\Psi=\bar s(1 - \bar s)(\rho - \bar s)(\rho + 1 - \bar s)+{\bar \alpha}^2\psi_1+\bar \alpha\psi_2,$ where $\psi_1$ and $\psi_2$ are given in \eqref{psi1psi}.

Substituting  ${\bar v}^s_{q0}, \forall q\in\mathcal{Q}$, in \eqref{vsnonpreemptive} into \eqref{MGFage}
results in the MGF of the AoI of source 1 under the blocking policy, given in Theorem~\ref{cmgf2}.

\section{Validations and Numerical Results}\label{Numerical Results}
By using the derived MGFs the $m$th moment of the  AoI of source 1, $\Delta_1^{(m)}$, is given by
$
\Delta_1^{(m)}=\dfrac{\mathrm{d}^m(M_{\Delta_1}(s))}{\mathrm{d}s^m}\Big|_{s=0}.
$ It can be shown that the average AoI under the preemptive policy derived by using the MGF generalizes the existing results in \cite{8406966} and \cite{8469047}. Specifically, 
when we set $\lambda_2=0,$ the average AoI follows the result in \cite{8406966}; and as $\alpha\rightarrow\infty$ or $\mu\rightarrow\infty,$ the average AoI coincides with the  average AoI expression under the LCFS-S policy introduced in \cite{8469047}.

 Fig. \ref{SD}  illustrates  the average AoI of source 1 and its standard deviation ($ \sigma $) as a function of $\lambda_1$ under the considered policies for $\mu=1$, $\alpha=1$, and $\lambda=5$. Simulated curves are shown as a verification of the derived theoretical results.  From the figure, it can be seen that the preemptive policy outperforms the blocking policy. 
 We would like to point out that 
 based on our experiments (not included here due to space limitation), this conclusion appears to hold true for any choice of the parameters. 
 In addition, it can be seen that the first moment ($\Delta_1$) and the standard deviation ($\sigma$) of the AoI increase by the same amount as $\rho_1$ decreases; thus their difference remains approximately constant.
 
 
 
 
 \begin{figure}
 	\centering
 	\includegraphics[width=1.1\linewidth,trim = 15mm 0mm 0mm 0mm,clip]{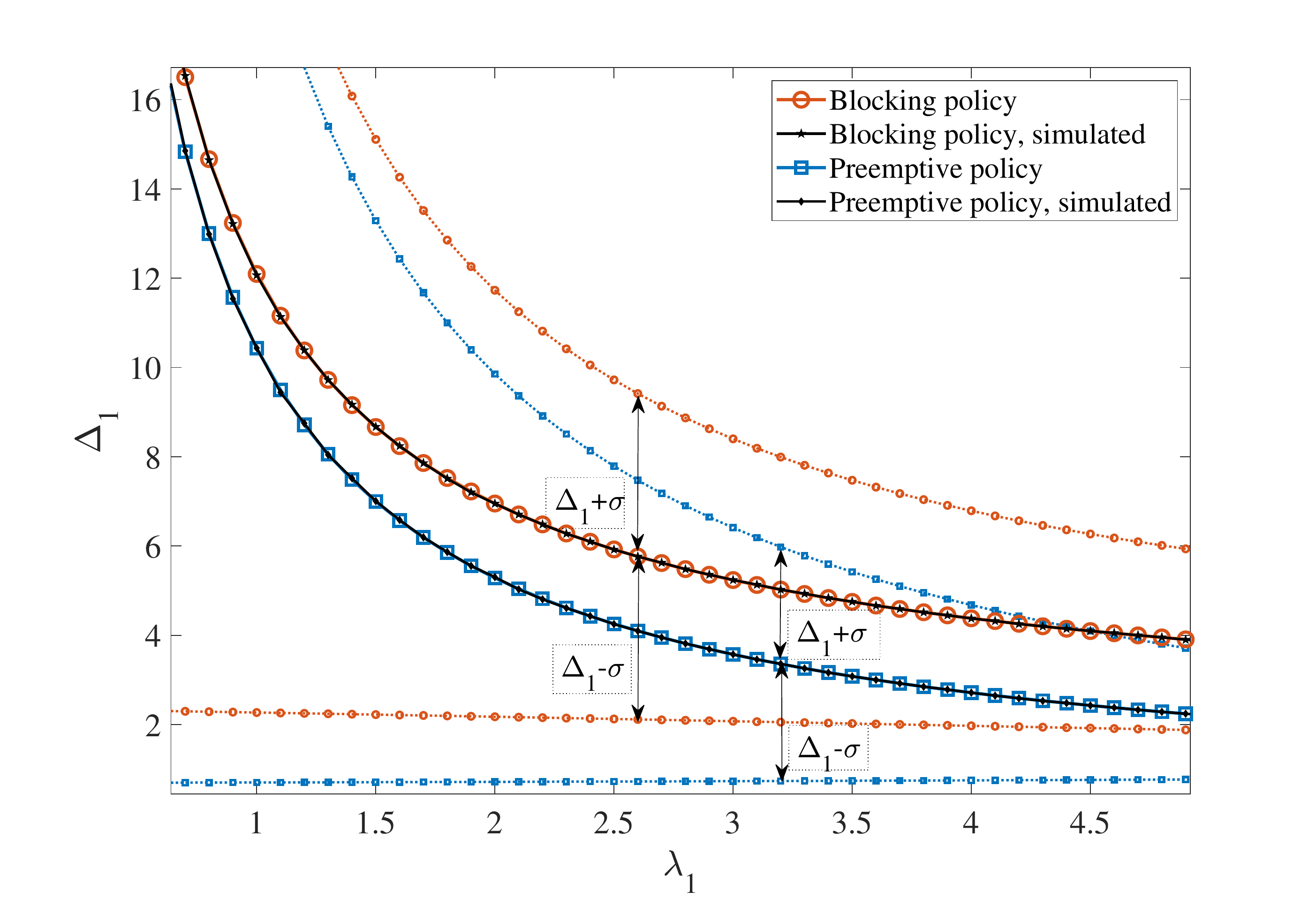}			
 	\caption{The average AoI of source 1 and its standard deviation ($\sigma$) as a function of $\rho_1$ under the two packet management policies for $\mu=1$, $\alpha=1$, and $\lambda=5$. The dashed lines visualize the standard deviation of the AoI as $\Delta_1+\sigma$ and $\Delta_1-\sigma$.}
 	\vspace{-5mm}
 	\label{SD}
 \end{figure}

\section{Conclusions}\label{Conclusions}
We considered a multi-source  system with computation-intensive status updates and derived the MGF of the AoI under two packet management policies using the SHS technique. Using the derived MGFs we evaluated  the average AoI  and its standard deviation. 
The results showed that the preemptive policy provides better performance than that of the blocking policy.

\appendices
\section{Values of $\bar{v}^s_{q0}$ for  the Preemptive Policy} \label{Valuesof  mgf preemptive}
\begin{align}\nonumber
&\bar{v}^s_{00}=\\&\nonumber\dfrac{\rho_1\bar\alpha(\bar s- \bar\alpha)(1-\bar s)(\bar\alpha+\rho+1-\bar s)}{(\bar\alpha\!+\!\rho\!-\!\bar s)(1\!+\!\rho\!-\!\bar s)(\bar s(1\! -\! \bar s)(\rho\! -\! \bar s)(\rho\! +\! 1 \!-\! \bar s)\!+\!{\bar \alpha}^2\psi_1\!+\!\bar \alpha\psi_2)},
\\&\nonumber
\bar{v}^s_{10}=\dfrac{\bar{v}^s_{00}\rho_1\hat\Psi}{(\bar\alpha + \rho - \bar s)(1-\bar s)(\bar\alpha+\rho+1-\bar s)},
\\&\nonumber
\bar{v}^s_{20}=\dfrac{\bar{v}^s_{00}\rho_2\hat\Psi}{(\bar\alpha + \rho - \bar s)(1-\bar s)(\bar\alpha+\rho+1-\bar s)},
\\&\nonumber
\bar{v}^s_{30}=\dfrac{\bar{v}^s_{00}\rho_1(\bar\alpha+\rho+1-2\bar s)}{(\bar\alpha + 1 - \bar s)(1-\bar s)(\bar\alpha+\rho+1-\bar s)},
\\&\nonumber
\bar{v}^s_{40}=\dfrac{\bar{v}^s_{00}\rho_1^2(\bar\alpha+\rho+1-2\bar s)}{(\bar\alpha + 1 - \bar s)(\bar \alpha- \bar s)(1-\bar s)(\bar\alpha+\rho+1-\bar s)},
\\&\nonumber
\bar{v}^s_{50}=\dfrac{\bar{v}^s_{00}\rho_1\rho_2(\bar\alpha+\rho+1-2\bar s)}{(\bar\alpha + 1 - \bar s)(\bar \alpha- \bar s)(1-\bar s)(\bar\alpha+\rho+1-\bar s)},
\\&\nonumber
\bar{v}^s_{60}=\dfrac{\bar{v}^s_{00}\rho_2(\bar\alpha+\rho+1-2\bar s)}{(\bar \alpha- \bar s)(1-\bar s)(\bar\alpha+\rho+1-\bar s)},
\\&\nonumber
\bar{v}^s_{70}=\dfrac{\bar{v}^s_{00}\rho_1\rho_2(\bar\alpha+\rho+1-2\bar s)}{(\bar\alpha + 1 - \bar s)(\bar \alpha- \bar s)(1-\bar s)(\bar\alpha+\rho+1-\bar s)},
\\&\nonumber
\bar{v}^s_{80}=\dfrac{\bar{v}^s_{00}\rho_2^2(\bar\alpha+\rho+1-2\bar s)}{(\bar\alpha + 1 - \bar s)(\bar \alpha- \bar s)(1-\bar s)(\bar\alpha+\rho+1-\bar s)},
\end{align}
where $\hat\Psi=\bar \alpha^2+\bar\alpha(2\rho-2-3\bar s)+\rho^2+\rho_1(2\rho_2+2-3\bar s)+\rho_2^2+\rho_2(2-3\bar s)+1-3\bar s+2\bar s^2$, and  $\psi_1$ and $ \psi_2 $ are given in \eqref{psi1psi}.

\bibliographystyle{IEEEtran}
\bibliography{Bibliography}
\end{document}